\documentclass[aps,pra,amsmath,amssymb,showpacs,twocolumn]{revtex4-1}

\usepackage[T1]{fontenc}
\usepackage{amssymb,amsmath,bm,bbold}
\usepackage{amsfonts}
\usepackage{graphicx}
\usepackage{amsthm}
\usepackage{color}
\usepackage{mathbbol}
\usepackage{epstopdf}
\usepackage{mathrsfs}
\usepackage{placeins}
\usepackage{tikz}
\usepackage{verbatim}

\newcommand{\vp}{\ensuremath{\varphi}}
\newcommand{\ii}{\ensuremath{ \mathrm{i\,} }}

\definecolor{amber(sae/ece)}{rgb}{1.0, 0.49, 0.0}

\newtheorem{lemma}[]{Lemma}

\newcommand{\ket}[1]{\ensuremath{\vert \, #1 \, \rangle}} 
\newcommand{\bra}[1]{\ensuremath{\langle \, #1 \, \vert}} 
\newcommand{\pscal}[2]{\ensuremath{ \langle \, #1 \, \vert  \, #2 \, \rangle}} 
\newcommand{\ppscal}[1]{\pscal{#1}{#1}} 
\newcommand{\dens}[2]{\ensuremath{ \vert \, #1 \, \rangle \langle \, #2 \, \vert}} 
\newcommand{\ddens}[1]{\dens{#1}{#1}} 
\newcommand{\moy}[1]{\ensuremath{ \langle \, #1 \, \rangle}} 
\newcommand{\moyvec}[2]{\ensuremath{ \langle #2 \, | \, #1 \, | \, #2 \rangle}} 

\newcommand{\com}[2]{\ensuremath{ [  #1 , #2  ] }}

\newcommand{\pnorm}[2]{\ensuremath{ \Vert  #1  \Vert_{#2} }} 

\newcommand{\id}{\ensuremath{ \mathcal{I} }}
\newcommand{\idc}{\ensuremath{ \mathrm{Id} }}

\newcommand{\tr}[1]{\ensuremath{  \operatorname{tr}\! \left[ #1 \right]  }}

\newcommand{\Vstate}[2]{\ensuremath{\operatorname{Var}[#1,#2]}}

\newcommand{\spn}[1]{\operatorname{Span} \left\lbrace #1 \right\rbrace }

\newcommand{\qfi}[2]{\mathop{}\! I(#1 \,;#2)}
\newcommand{\cqfi}[2]{\mathop{}\! C(#1 \,;#2)}

\newcommand{\e}[1]{\ensuremath{  \operatorname{e}^{#1}}}

\newcommand{\Vcl}[1]{\operatorname{Var}[#1]}

\newcommand{\est}[1]{\hat{#1}_{\mathrm{est}}  }

\newcommand{\cfi}[2]{\mathop{}\! J(#1 \,;#2)}

\newcommand{\ens}[1]{\ensuremath{ {\left\{ #1 \right\}}} }

\newcommand{\ve}{\ensuremath{\varepsilon}}

\newcommand{\mc}[1]{\mathcal{#1}}

\newcommand{\msc}[1]{\mathscr{#1}}

\newcommand{\diff}[1][]{\mathop{}\!\mathrm{d_{#1}}}
\newcommand{\frd}[3][]{%
  \frac{%
   \diff \ifx\\#1\\\else^{#1}\fi #2
  }{%
    \diff #3 \ifx\\#1\\\else^{#1}\fi
  }%
}
\newcommand{\frp}[3][]{%
  \frac{%
    \partial \ifx\\#1\\\else^{#1}\fi #2
  }{%
    \partial #3 \ifx\\#1\\\else^{#1}\fi
  }%
}

\newcommand{\crr}{}
\def\crr/{Cram\'er-Rao}
\newcommand{\csch}{}
\def\csch/{Cauchy–Schwarz}

\begin{document}

\title{Enhancing sensitivity in quantum metrology by
  Hamiltonian extensions} 

\author{Julien Mathieu Elias Fra\"isse$^{1}$  and Daniel Braun$^{1}$ }
\affiliation{$^1$ Eberhard-Karls-Universit\"at T\"ubingen, Institut
  f\"ur Theoretische Physik, 72076 T\"ubingen, Germany}

\begin{abstract} 
A well-studied  
scenario in quantum parameter estimation theory 
arises when the parameter to be estimated is imprinted on the
  initial state by 
a Hamiltonian
  of the form $\theta G$.  For such ``phase shift
Hamiltonians'' it has been shown that one cannot improve the
channel quantum Fisher information
by adding ancillas and letting the system interact with them.  Here we
investigate the general case, where the  Hamiltonian is not necessarily
a phase shift, and show that in
this case in general
it \emph{is} possible to increase the quantum
channel information and to
reach an upper bound. This can be done 
by adding a term proportional to the derivative of the Hamiltonian, or
by subtracting a term to the original Hamiltonian.
 Neither method
makes use of any ancillas which shows that for 
quantum channel estimation with arbitrary parameter-dependent
Hamiltonian, entanglement with an ancillary 
system is not necessary to reach the best possible
sensitivity.  By adding an operator to the Hamiltonian we can also
modify the time scaling of the channel quantum Fisher information.  We
illustrate   
our techniques 
 with  NV-center magnetometry and the estimation of
the direction of a magnetic field in a given plane using a single
spin-1 as probe.  

\end{abstract}

\maketitle

\section{Introduction}\label{intro}

With the increasing ability of controlling quantum systems, quantum
metrology has become a major and lively field. One important task in
quantum metrology is the parameter estimation of quantum channels,
also known as quantum channel estimation.
It is important for at least two
reasons. Firstly, because to carry out experiments,  it is crucial  to
know exactly the processes applied to the system. Channel
estimation is then used as quantum process tomography. 
Secondly, quantum channel estimation is useful for  
investigating the utility of a system for measuring 
some relevant
physical quantity, and 
for increasing the sensitivity of
sensors.  

The field has attracted a lot of attention when it became clear that
for a given number of probes 
using entanglement between the probes 
can 
increase the sensitivity compared to the classical case \cite{vittorio_giovannetti_quantum_enhanced_2004,giovannetti_quantum_2006}. 
A second possibility of using entanglement  is to introduce an ancillary
system and to  entangle it with the original system, while
still applying the
quantum channel only to the original system. 
It was shown that such channel
extensions can 
sometimes enhance the 
sensitivity
\cite{fujiwara_quantum_2001,dariano_using_2001}, in the sense of
increasing  the Quantum Fisher Information (QFI).  

Among quantum channels, 
unitary channels play a particular role. In these, the time evolution
is obtained through propagation with a unitary operator, obtained  through the 
Hamiltonian  that depends on the parameter to be estimated, 
and quantum channel estimation for unitary channels is therefore also
known as Hamiltonian parameter estimation. 
Phase-shift
Hamiltonians correspond to the special case where the parameter to be estimated 
multiplies a Hermitian generator.
The corresponding
parameter estimation problem is particularly relevant due to the
importance of phase measurements in physics. A typical example of such
a situation is the estimation of a phase in an
interferometer. For unitary channels, we can go beyond 
channel extension by adding  an arbitrary
parameter-independent operator
to the Hamiltonian, which may include 
even interactions with 
ancillary systems that can be initially entangled with the original
system.  Adding a parameter-independent Hamiltonian to the original
Hamiltonian is known as ``Hamiltonian extension''
\cite{fraisse_hamiltonian_2016}.
For  
phase shifts, Hamiltonian extension does not improve the best possible
sensitivity.  
This was shown in
\cite{boixo_generalized_2007,fraisse_hamiltonian_2016} by calculating
an upper bound for the 
QFI for unitary channels and showing that the bound is saturated for
phase shift Hamiltonians.

Here we investigate the
 problem of saturating the upper bound calculated in 
 \cite{boixo_generalized_2007,fraisse_hamiltonian_2016} for general
 Hamiltonians, where it is typically not saturated yet,
by extending the
 Hamiltonian. 
Indeed, from a pure metrological point of view it is 
 interesting to check 
whether or not one can go beyond the optimization over initial state
and POVM measurement
by engineering also the Hamiltonian to increase the 
sensitivity, without introducing additional parameter
dependence. 
In particular, one would like to know whether  ancilla-assisted
schemes can provide an 
advantage for general Hamiltonian parameter estimation. We find
that 
for general parameter dependent Hamiltonians, the
sensitivity can be increased by adding a part to the Hamiltonian
given by its local derivative with respect to the parameter,
multiplied with a large, parameter-independent prefactor.
 This can be
understood intuitively as a ``signal flooding'', \emph{i.e.}~the
relative weight of 
the useful part of the Hamiltonian is enhanced. A second  way to
  saturate the upper bound, is to subtract the Hamiltonian taken at a
  fixed value of the parameter from the original Hamiltonian. At this
  specific value of the parameter the new Hamiltonian saturates the
  upper bound.  Neither scheme needs any ancillas. 

A third opportunity for Hamiltonian extension exists if the
eigenvalues of the Hamiltonian are independent of the parameter.  This
leads to a   periodic time-dependence of the channel QFI.  An
addition to the Hamiltonian that breaks the parameter-independence
of the eigenvalues can create a quadratic time
  scaling in the channel QFI and hence for large enough times
  strongly increase the sensitivity.     
While typically this method does not 
  saturate the upper bound, it is still useful for metrology as it
  allows one to use 
time as a resource.  {The engineering of the time scaling to
  obtain the $t^2$ scaling was also shown in the context of metrology
  with feedback controls
  \cite{yuan_optimal_2015,yuan_sequential_2016}. }

We illustrate our results with NV-center
magnetometry, where we find that optimal sensitivity can be reached by
adding a strong external magnetic field in the direction of the field
to be measured, {\em i.e.}~flooding the signal 
with a strong
known signal of the same kind.  No ancillas are required to reach
optimal sensitivity. We also study  the qualitatively different estimation of a direction of
  a magnetic field with a single spin-1 as a probe 
which  allows us to compare 
the three different methods.

\section{Quantum metrology}

\subsection{Channel QFI}

In parameter estimation theory, 
the goal is to infer the value of a parameter $\theta$ given
the realization of an $m$-sample  of a random variable whose
distribution depends on $\theta$. This is done by an estimator
$\est{\theta}$. A common property required for estimators is to be 
unbiased, meaning that in an infinitesimal interval about the true
value of the parameter $\theta$, on average the estimator should give
that true value of the parameter, $\moy{\est{\theta}}=\theta$. A second
desirable property is to have a variance as small as possible, such
that the estimate fluctuates as little as possible about the true value of the 
parameter. 

Parameter estimation theory applies naturally to the metrology of
quantum systems. 
The most general measurements correspond to
POVMs (Positive Operator Valued Measure), \emph{i.e.}  sets
$\ens{E_\xi}$ of positive semi-definite operators fulfilling the
closure relation $\sum_\xi E_\xi =\id$.   Given a state $\rho_\theta$,
a POVM generates a probability distribution $\mu_\theta(\xi) =
\tr{E_\xi \rho_\theta}$, mapping the problem of estimating a parameter
of a quantum state to the problem of parameter estimation. This
approach leads to a fundamental theorem, known as the quantum \crr/
theorem \cite{helstrom_quantum_1969,braunstein_statistical_1994}, which sets a
bound on the variance of any unbiased estimator 
given a state $\rho_\theta$,
\begin{equation}
\Vcl{\est{\theta}} \geq \frac{1}{\qfi{\rho_\theta}{\theta}}\;,
\end{equation}
with $\qfi{\rho_\theta}{\theta}$ the Quantum Fisher Information (QFI),
defined as $\qfi{\rho_\theta}{\theta} = \tr{\rho_\theta
  L_\theta^2}$. The so-called symmetric logarithmic derivative
$L_\theta$ is defined implicitly by $\frp{\rho_\theta}{\theta}
=(\rho_\theta L_\theta+L_\theta \rho_\theta)/2$.  The 
bound results from 
a double optimization: a first optimization
over all unbiased estimators, and then a second optimization over all
possible POVMs. When repeating the measurement $m$ times independently
---  which amounts to generating an $m$-sample --- the quantum \crr/
theorem reads $\Vcl{\est{\theta}} \geq
(m\qfi{\rho_\theta}{\theta})^{-1}$. The bound can always be saturated
in the limit of large numbers $m$ of measurements by using the maximum
likelihood estimator and building the POVM based on the eigenvectors
of $L_\theta$.

Often the parameter to be estimated characterizes 
a physical process, \emph{i.e.}~a quantum 
channel: 
One starts with an initial state $\rho_0$
of the probe that is independent of the parameter. 
Then we let the quantum channel $\mc{E}_\theta$ act on the probe,
giving as a result the state $\rho_\theta = \mc{E}_\theta (\rho_0)$.  
A POVM measurement is then performed, whose outcome is fed into
$\hat{\theta}_\text{est}$ and provides us with an estimate of $\theta$.
We see that in this case we have a new degree of freedom in the QFI,
the choice of the initial state.  One thus introduces a new quantity,
the channel QFI $C$, which corresponds to the largest QFI
reachable for a given channel, 
\begin{equation}
\cqfi{\mc{E}_\theta}{\theta}=\max_{\rho_0}  \qfi{\mc{E}_\theta(\rho_0)}{\theta}\;.
\end{equation}
Importantly, due to the convexity of the QFI, it is enough to maximize
over pure states, $\rho_0=\ddens{\psi_0}$. 

\subsection{Channel extensions and Hamiltonian extensions}
We now turn our attention to extensions for metrology.
It has been shown that by adding an ancilla
to the probe but still acting only with the channel on the probe,
\emph{i.e.}~applying $\mc{E}\otimes \idc$ to the whole system, one
can, for certain channels, improve the channel QFI
\cite{dariano_using_2001,fujiwara_quantum_2001}. We call such 
extensions ``channel extensions'', in contrast to ``Hamiltonian
extensions'', defined below. For 
unitary channels, channel extensions 
 do not increase the
channel QFI: $\cqfi{\mc{U}_{{H(\theta)}} \otimes
  \idc}{\theta}=\cqfi{\mc{U}_{{H(\theta)}}}{\theta}$  (see
section \ref{sec.cqfi} for a short proof). 

But when extending the Hamiltonian $H(\theta)$
by adding to it another parameter-independent 
Hamiltonian $H_1$, 
 \begin{equation}\label{eq:def_ext_ph_shift}
H_{\mathrm{ext}}(\theta)= 
H(\theta)\otimes \id + H_{1}\;, 
\end{equation}
we are not anymore in the situation covered by channel
extension.
Eq.\eqref{eq:def_ext_ph_shift} is our formal definition of a ``Hamiltonian
extension''. $H_\text{1}$ can be a coupling between the original 
system and an ancilla, but can also refer to a Hamiltonian that acts
non-trivially only on the Hilbert space of the original system or on the
Hilbert space of the ancilla. 
Can such a Hamiltonian
extension lead to an improvement in the channel QFI? It was shown
in  \cite{fraisse_hamiltonian_2016} that for the specific case of
phase shift Hamiltonians, \emph{i.e.} Hamiltonians of the  form
$H(\theta) = \theta G$,  such Hamiltonian extensions cannot improve the
channel QFI, a 
result that was
already known in the context of
many-body interaction metrology 
\cite{boixo_generalized_2007}.  
Here we investigate the question more generally for arbitrary Hamiltonians
$H(\theta)$. 

\section{Metrology with unitary channels}

\subsection{Channel QFI {and semi-norm} }\label{sec.cqfi}

From now and for the rest of the article we will focus on unitary
channels. These channels correspond to the unitary evolution of the
state of a closed
system
described by the Schr\"odinger equation. Given the Hamiltonian
$H(\theta)$, the effect of the channel $\mc{U}_{H(\theta)}$ on the
initial state $\rho_0$ is given by
$\mc{U}_{H(\theta)}(\rho_0)=U_{H(\theta)}
\rho_0{U_{H(\theta)}}^\dagger $ with $U_{H(\theta)}=\e{-\ii t
  H(\theta)}$ the evolution operator (we take $\hbar=1$  throughout
the paper apart from the section \ref{sec:NV}).  Since we only deal
with unitary channels we adopt  the notation
$\cqfi{\mc{U}_{H(\theta)}}{\theta} \equiv \cqfi{H(\theta)}{\theta} $
for the channel QFI. 
In general, calculating the QFI is a difficult task as it requires one
to diagonalize the density matrix. For pure states the expression is
still simple and reads,
$\qfi{\ddens{\psi_\theta}}{\theta}=\ppscal{\dot{\psi}_\theta}-\vert
\pscal{\psi_\theta}{\dot{\psi}_\theta} \vert^2$, where here and
throughout the article the dot stands for the derivative with respect
to the parameter to be estimated, {\em
  i.e.}~$\ket{\dot{\psi}_\theta}=\frp{\ket{\psi_\theta}}{\theta}$. Introducing
the local generator 
 \begin{equation}
   \msc{H} =  \ii {U_{H(\theta)}}^\dagger \dot{U}_{H(\theta)}\;,
\end{equation}
we can write the QFI as
\begin{equation}\label{eq:qfi_gen_ham_3}
\qfi{\ddens{\psi_\theta}}{\theta} =4\Vstate{\msc{H}}{\ket{\psi_0}}\;,    
\end{equation}
with $ \Vstate{\msc{H}}{\ket{\psi_0}} \equiv
\moyvec{\msc{H}^2}{\psi_0} - \moyvec{\msc{H}}{\psi_0}^2$, 
and consequently the channel QFI  as
\begin{equation}
\cqfi{H(\theta)}{\theta}= 4 \max_{\ket{\psi_0}\in\mc{H}} \Vstate{\mathscr{H}}{\ket{\psi_0}}\;.
\end{equation}
The maximization {\cite{giovannetti_quantum_2006}} 
 can be done as follows: Using Popoviciu's inequality
\cite{Popoviciu35}, which states that for a random variable $X$, with
minimal value $a$ and maximal value $b$, the variance of $X$ is upper
bounded by $(b-a)^2/4$, and then noticing that in
eq.\eqref{eq:qfi_gen_ham_3} the  variance saturates its upper bound
for states of the form $(\ket{h_{\mathrm{M}}}+\e{\ii \vp}
\ket{h_{\mathrm{m}}})/\sqrt{2}$ 
(where $\ket{h_{\mathrm{M}}}$ and $\ket{h_{\mathrm{m}}}$ correspond,
respectively, to  
eigenvectors  of $\msc{H}$ with maximal eigenvalue ($h_{\mathrm{M}}$) and minimal eigenvalue ($h_{\mathrm{m}}$)), 
 we have 
\begin{equation}
\cqfi{H(\theta)}{\theta}= (h_{\mathrm{M}}-h_{\mathrm{m}})^2\;.
\end{equation}

In order to simplify the calculation of the channel QFI, and following the method in \cite{boixo_generalized_2007}, we introduce the semi-norm
\begin{equation}
\pnorm{A}{sn} = a_{\mathrm{M}}-a_{\mathrm{m}}
\end{equation}
with $ a_{\mathrm{M}}$ and $ a_{\mathrm{m}}$ the maximal and minimal
eigenvalues of $A$ { (we call such eigenvalues \emph{extremal}
  eigenvalues, and their associated eigenvectors, \emph{extremal}
  eigenvectors ; we also call maximal (resp.~minimal) eigenvector an
    eigenvector corresponding to a maximal (resp.~minimal)
    eigenvalue). We then have the simple expression for the channel
  QFI 
\begin{equation}\label{eq:cqfi_sn}
\cqfi{H(\theta)}{\theta}= \pnorm{\msc{H}}{sn}^2\;.
\end{equation}

As it will be important in the following, let us show the triangle 
inequality for this semi-norm. Let $A=B+C$. Then
$\pnorm{A}{sn}=\moyvec{A}{a_{\mathrm{M}}} - 
\moyvec{A}{a_{\mathrm{m}}}=\moyvec{B}{a_{\mathrm{M}}}+\moyvec{C}{a_{\mathrm{M}}}
- \moyvec{B}{a_{\mathrm{m}}} - \moyvec{C}{a_{\mathrm{m}}}$, with $\ket{a_{\mathrm{M}}}$ (resp.~$\ket{a_{\mathrm{m}}}$)  a maximal (resp.~minimal) eigenvector {of $A$}.
By definition $\moyvec{B}{a_{\mathrm{M}}} \leq b_{\mathrm{M}}$ and
$-\moyvec{B}{a_{\mathrm{m}}} \leq -b_{\mathrm{m}}$ with $
b_{\mathrm{M}}$ and $ b_{\mathrm{m}}$ the maximal and minimal
eigenvalues of $B$. In the same way we have
$\moyvec{C}{a_{\mathrm{M}}} \leq c_{\mathrm{M}}$ and
$-\moyvec{C}{a_{\mathrm{m}}} \leq -c_{\mathrm{m}}$ with $
c_{\mathrm{M}}$ and $ c_{\mathrm{m}}$ the maximal and minimal
eigenvalues of $C$.  We thus get $\pnorm{A}{sn}=\pnorm{B+C}{sn} \leq
b_{\mathrm{M}}+c_{\mathrm{M}}-b_{\mathrm{m}}-c_{\mathrm{m}} =
\pnorm{B}{sn}+\pnorm{C}{sn}$ which is exactly the triangle
inequality. {When $B$ and $C$ have no degenerate extremal eigenvalues } the equality is reached for $\ket{b_{\mathrm{M}}} \propto
\ket{c_{\mathrm{M}}}$ and $\ket{b_{\mathrm{m}}} \propto
\ket{c_{\mathrm{m}}}$, meaning that both operators have to share the
same extremal eigenvectors.  For the degenerate case, the triangle
inequality is saturated 
if and only if the intersection of the invariant
subspaces of the maximal (resp.~minimal) eigenvalue of $B$ and of $C$
is not empty \footnote{Say $A$ is an operator acting on $E$. Then a
  subspace $M\subset E$ is an invariant subspace of $A$ if and only if
  $AM\subset M$. We also say that $M$ is stable by $A$}.  Stated
otherwise, $B$ and $C$ should share (in the sense of proportionality)
at least one {maximal  and one minimal eigenvector}. 
Since the eigenvalues 
are preserved by
similarity transformations, we also have $\pnorm{U A
  U^{-1}}{sn}=\pnorm{A}{sn}$ for any unitary $U$, regardless of whether
$U$ depends on $\theta$ or not. 

We can use eq.\eqref{eq:cqfi_sn} to show that channel extension does
  not increase the channel QFI of unitary channels: For a unitary
  channel $\mc{U}_{H(\theta)}$ the extended channel
  $\mc{U}_{H(\theta)}\otimes \idc$ can be written as
  $\mc{U}_{H(\theta)\otimes \id}$. We furthermore have
  ${U_{H(\theta)\otimes \id}}^\dagger \dot{U}_{H(\theta)\otimes
    \id}={U_{H(\theta)}}^\dagger \dot{U}_{H(\theta)} \otimes \id$
  showing that $\cqfi{\mc{U}_{H(\theta)}\otimes \idc}{\theta}=
  \pnorm{\msc{H}\otimes
    \id}{sn}=\pnorm{\msc{H}}{sn}=\cqfi{\mc{U}_{H(\theta)}}{\theta}$,
  using eq.\eqref{eq:cqfi_sn} plus the fact that $\msc{H}\otimes \id$
  has the same eigenvalues as $\msc{H}$.

\subsection{Upper bound for the channel QFI} \label{sec:upper_bound}

\begin{lemma}[Upper bound for channel QFI \cite{boixo_generalized_2007}]\label{Boixo}
For general Hamiltonians $H(\theta)$,  with associated evolution operator $U_{H(\theta)}=\e{-\ii t H(\theta) }$, the channel QFI  is upper bounded {as}
\begin{equation}\label{eq:main_ineq}
\cqfi{ H(\theta) }{ \theta }  \leq t^2 \pnorm{\dot{H}(\theta)}{sn}^2\;.
\end{equation}
 \end{lemma}

\begin{lemma}[Saturation of the bound]\label{sat}
In the case where $\dot{H}(\theta)$ has no degenerate extremal
eigenvalues, 
equality in \eqref{eq:main_ineq} is reached if and
only if the extremal eigenvectors of $\dot{H}(\theta)$ are also
eigenvectors  of $H(\theta)$.  

In the degenerate case, 
equality in \eqref{eq:main_ineq} is reached if and
only if  there exist $\ket{\psi} \in \mc{P}$ and $\ket{\phi} \in \mc{D}$ such that 
$V(\alpha) \ket{\psi}  \in \mc{P}$ and $V(\alpha) \ket{\phi}  \in
\mc{D}$ for all $\alpha \in [0,1]$ with
$\mc{P}=\spn{\ket{1},\cdots,\ket{a}}$
(resp.~$\mc{D}=\spn{\ket{b},\cdots,\ket{d}}$) 
the invariant subspace of $\dot{H}(\theta)$  associated with its
maximal 
(resp.~minimal) eigenvalue, and where $V(\alpha)=\e{-\ii \alpha t
  H(\theta)}$. A sufficient condition is that there exists an
eigenvector of $H(\theta)$ 
 in $\mc{P}$ and another one in $\mc{D}$. 
 
The inequality is saturated by
   $\ket{\psi_\mathrm{opt}} =(\ket{M}+\ket{m})/\sqrt{2}$ with $\ket{M}
   \in \mc{P}$ and $\ket{m} \in \mc{D}$, {\em i.e.}~a balanced
   superposition of a maximal eigenvector and a minimal eigenvector of
   $\dot{H}(\theta)$.
\end{lemma}

\begin{proof}[Proof of Lemma \ref{Boixo}\cite{boixo_generalized_2007}] 
In general \cite{wilcox_exponential_1967,snider_perturbation_1964}, for a matrix $M(x)$ depending on the parameter $x$, the derivative of its exponential with respect to $x$ is given by 
\begin{equation}\label{eq:der_exp}
\frd{\e{M(x)}}{x} =\int_{0}^1 \e{\alpha M(x)}\frd{M(x)}{x}\e{(1-\alpha) M(x)}\diff \alpha\;.
\end{equation}
From this expression we can re-express the local generator of the translation,
\begin{equation}
\msc{H} = t \int_{-1}^0 W(\alpha,\theta) \diff \alpha\;,
\end{equation}
with $W(\alpha,\theta) =V(\alpha)\dot{H}(\theta) V(\alpha)^\dagger$,
where $V(\alpha)=\e{-\ii \alpha t H(\theta)}$.
Applying the triangle inequality to the semi-norm of $\msc{H}$, and noticing that $\pnorm{W(\alpha,\theta)}{sn}=\pnorm{\dot{H}(\theta)}{sn}$ we obtain 
\begin{equation} \label{eq:first_ineq}
\pnorm{\msc{H}}{sn} \leq  t \int_{-1}^0 \pnorm{ W(\alpha,\theta) }{sn}\diff \alpha = t\pnorm{\dot{H}(\theta)}{sn}\;.
\end{equation}
\end{proof}

\begin{proof}[Proof of Lemma \ref{sat}]
 Since $H(\theta)$ is Hermitian, so is
$\dot{H}(\theta)$ and we can then write it in its orthonormal eigenbasis 
as $\dot{H}(\theta) =\sum_{i=1}^d e_i \ddens{i}$, with $e_1=\cdots=e_a > e_i > e_{b}=
\cdots =e_d 
$ for all $a<i<b$, where $a$ is the dimension of the
invariant subspace of maximal eigenvalue and $d-b+1$ is the dimension
of the invariant subspace with minimal eigenvalue.
Since $W(\alpha,\theta)$ is related to
$\dot{H}(\theta)$ by a similarity transformation, we can write it as
$W(\alpha,\theta) = \sum_{i=1}^d e_i 
\ddens{\alpha_i}$ with $\ket{\alpha_i} = V(\alpha) \ket{i}$. The
condition for equality in eq.\eqref{eq:first_ineq} and hence in
  eq.\eqref{eq:main_ineq} is that there exists a vector $\ket{\psi}$
which is an eigenvector of $W(\alpha,\theta)$  with eigenvalue $e_1$
simultaneously for all values of $\alpha \in [-1,0]$, and a vector
$\ket{\phi}$ which is an eigenvector of $W(\alpha,\theta)$ with
eigenvalue $e_d$ simultaneously for all values of $\alpha \in
[-1,0]$. 

Consider first the case where
$\dot{H}(\theta)$ has no degenerate extremal eigenvalues,
\emph{i.e.}~$a=1$ and $b=d$.  
 Then the condition for equality 
is equivalent to 
$\ket{1}$
and $\ket{d}$ being 
eigenvectors of $V(\alpha)$ for all $ \alpha \in [-1,0]$.  
Let us see how we can re-express this condition so that it
only involves {eigenvectors} of 
 $\dot{H}(\theta)$ and $H(\theta)$. By expressing the
Hamiltonian in its eigenbasis, $H(\theta) = \sum_i h_i \ddens{h_i}$,
we can write $V(\alpha)=\e{-\ii \alpha t h_i}\ddens{h_i}$. 
The eigenvectors of $H(\theta)$ are also eigenvectors of $V(\alpha)$ but
not the other way round: Indeed, $\e{-\ii \alpha t h_i}$ may be equal to  $\e{-\ii \alpha t h_j}$
while $h_i \neq h_j$. In such cases we can construct eigenvectors of
$V(\alpha)$ which are not eigenvectors of $H(\theta)$ by linearly
combining $\ket{h_i}$ and $\ket{h_j}$. Nevertheless this  
can happen only for a countable number of $\alpha$-values 
 given $t$, whereas in all other
cases the eigenvectors of $V(\alpha)$ must also be eigenvectors of
$H(\theta)$. Hence, the condition $\e{-\ii \alpha t h_i}=\e{-\ii \alpha t h_j}$ cannot be satisfied for all
$ \alpha \in [-1,0]$ if $h_i \neq h_j$,
and therefore the condition for equality in \eqref{eq:first_ineq} can
be stated as: $\ket{1}$ and $\ket{d}$ should be eigenvectors of
$H(\theta)$.

In the degenerate case  the 
eigenspace of $W(\alpha,\theta)$ with
eigenvalue $e_1$ is spanned by $\lbrace
V(\alpha)\ket{1},\cdots,V(\alpha)\ket{a} \rbrace$. The existence of a
common eigenvector of $W(\alpha,\theta)$ for $\alpha \in [-1,0]$ with
eigenvalue $e_1$  means  that there should exist a vector
$\ket{\psi}$ independent of $\alpha$ that can be written as
  $\sum_{i=1}^a \psi_i(\alpha) V(\alpha) \ket{i} =V(\alpha)
  \sum_{i=1}^a \psi_i(\alpha) \ket{i} =V(\alpha)
  \ket{\varphi(\alpha)}$ where the last equality defines
$\ket{\varphi(\alpha)}$. Since $\ket{\varphi(\alpha)} \in
  \mc{P}$ this is equivalent to say that there  should exist a vector
$\ket{\psi}$ such that $ V(-\alpha)\ket{\psi}$ belongs to $\mc{P}$ for
$\alpha \in[-1,0]$. This is true especially for $\alpha=0$ showing
that also $\ket{\psi}$ belongs 
to $\mc{P}$. A similar treatment
for the lowest eigenvalue shows that the necessary and sufficient
condition for equality in eq.\eqref{eq:main_ineq} 
is equivalent to
the existence of a vector $\ket{\psi} \in 
\mc{P}$ such that $V(\alpha)\ket{\psi} \in \mc{P}$ for   $\alpha
\in[0,1]$, and of a vector $\ket{\phi}\in \mc{D}$ such that
$V(\alpha)\ket{\phi} \in \mc{D}$ for $\alpha \in[0,1]$. 
  The sufficient condition is found by observing that the above condition is fulfilled if $H(\theta)$ has an eigenvector in $\mc{P}$ and an eigenvector in $\mc{D}$.  
\end{proof}

In the case of phase shifts, the condition for equality in
eq.\eqref{eq:main_ineq} 
is fulfilled as $\dot{H}(\theta)$ and  $H(\theta)$ are simultaneously diagonalizable, showing that for phase
shifts $\cqfi{\theta G}{\theta}=\pnorm{G}{sn}$. The Hamiltonian
extension $G_{\mathrm{ext}}(\theta)$ of a phase shift
is not a phase shift anymore.  We thus have
$\cqfi{G_{\mathrm{ext}}(\theta)}{\theta} \leq
\pnorm{\dot{G}_{\mathrm{ext}}(\theta)}{sn}=\pnorm{G \otimes
  \id}{sn}=\pnorm{G }{sn} = \cqfi{\theta G}{\theta}$, showing that
quantum metrology with phase shift Hamiltonians cannot profit from
Hamiltonian extension, in addition to not being able to profit from channel extensions.

\section{Saturating the bound} \label{sec:saturating_bound}

We have seen that  Hamiltonian extension fails to provide an advantage
in terms of channel QFI for  phase shift Hamiltonians. Nevertheless
the question of the order between $\pnorm{\mathscr{H}}{sn}$ and
$\pnorm{\mathscr{H}_{\mathrm{ext}}}{sn}$ is still open
for arbitrary Hamiltonians $H(\theta)$, {where
 $\mathscr{H}_{\mathrm{ext}}=\ii U_{H_{\mathrm{ext}}}{}^\dagger \dot{U}_{H_{\mathrm{ext}}} $ is the local generator of the extended Hamiltonian with $U_{H_{\mathrm{ext}}}=\e{-\ii t H_{\mathrm{ext}}}$.}
We only have the two inequalities
\begin{align}
\pnorm{\mathscr{H}}{sn}^2 &\leq t^2 \pnorm{\dot{H}(\theta)}{sn}^2 \label{upper_bound_1}\\
\pnorm{\mathscr{H}_{\mathrm{ext}}}{sn}^2 &\leq t^2 \pnorm{\dot{H}(\theta)}{sn}^2 \label{upper_bound_2}  \,,
\end{align}
where the second one follows from the fact that $H_\text{1}$ is
independent of $\theta$. 
The interesting question here is thus whether for 
a given $H(\theta)$ that does not saturate \eqref{upper_bound_1} 
we can saturate the bound \eqref{upper_bound_2} by tuning the interaction
Hamiltonian  in $H_{\mathrm{ext}}$, and therefore increase the
sensitivity, \emph{i.e.}~have
$\pnorm{\mathscr{H}_{\mathrm{ext}}}{sn}^2>\pnorm{\mathscr{H}}{sn}^2$. 
To answer this question, we first look at the specific case of {what
  we call} ``broken phase-shift'' 
before treating the general case. 

\subsection{Restoring a broken phase shift}\label{sec:restoring}

Let us consider the Hamiltonian $K(\theta) = \theta G + F$,
along with the corresponding unitary operator $U_{K(\theta) }=\e{-\ii K(\theta)t}$
and the corresponding local generator $\mathscr{K}=\ii
U_{K(\theta)}{}^\dagger\dot{U}_{K(\theta)}$. For $K(\theta)$ we have
from eq.\eqref{eq:main_ineq}
$\pnorm{\mathscr{K}}{sn}^2 \le t^2 \pnorm{G}{sn}^2$.
We assume that the conditions for equality in eq.\eqref{eq:main_ineq}
are not fulfilled and we therefore have $\pnorm{\mathscr{K}}{sn}^2 <
t^2\pnorm{G}{sn}^2$. Our goal is to design a Hamiltonian extension $K_{\mathrm{ext}}(\theta)$  for $K(\theta)$ such that the channel QFI  for  $K_{\mathrm{ext}}(\theta)$ 
saturates inequality \eqref{eq:main_ineq}, 
i.e.
\begin{equation}
\cqfi{K_{\mathrm{ext}}(\theta)}{\theta} =t^2 \pnorm{G}{sn}^2
\;.\label{16} 
\end{equation}
In order to saturate the bound two solutions appear directly.
 Loosely
speaking, we can either  cancel the part that spoils the Hamiltonian,
$ F$, or increase the useful part of  the Hamiltonian, $\theta
G$. Both of them correspond to a Hamiltonian extension but without
ancillas, just adding an extra part to the Hamiltonian. In the first
case the extra Hamiltonian is $ -  F$, and the {corresponding extended
  Hamiltonian} becomes 
 the phase shift $\theta G${, which saturates the bound 
\eqref{16}}. 
While this method may appear artificial,
it clearly demonstrates 
the possibility of enhancing the sensitivity in the case where
the parameter is not coded in a simple phase shift. In fact, it is
not even necessary to add the full $-F$. Assuming for simplicity
non-degenerate extremal eigenvalues of {$\dot{K}(\theta)$}, 
 we know from
Lemma \ref{sat} that only the 
extremal eigenvectors of {
$\dot{K}(\theta)$} have to be eigenvectors of
$K(\theta)$. With this one shows easily that a corrected
$\tilde{K}(\theta)\equiv K(\theta)+R$ {saturates the upper bound} 
if and only if 
$\bra{m} R \ket{n}=-\bra{m} F \ket{n}$ for $n=1$ with $m \in\{2,\cdots
,d\}$ and for 
$n=d$ with $m\in\{1,\cdots, d-1\}$}. {\em I.e.}~loosely speaking, $F$ needs
to be substracted only in 
the subspace of the extremal eigenvectors.

The second strategy is to
add a Hamiltonian $\beta G$. The extended Hamiltonian in this case
reads $K_{\mathrm{ext}}(\theta,\beta) =( \theta +\beta)G +  F$, 
 \emph{i.e.} we have
$K_{\mathrm{ext}}(\theta,\beta)=K(\theta+\beta)$. 
Hence, for large $\beta$, the eigenvectors of $K_\text{ext}$ become
those of $G$, {\em i.e.}~$K_\text{ext}$ and $\dot{K}_\text{ext}$ have then
the same extremal eigenvectors and the bound {\eqref{16}} 
 can be saturated. 
Importantly, this transformation does not
correspond to just
a re-parametrization: In a re-parametrization we keep
the same probability distribution, and we only change what we consider
to be the parameter of the probability distribution. In the present
case 
we change the probability distribution, but
we keep the original parameter, since it corresponds to the physical quantity
in which we are interested.  

Let us formalize this in terms of Fisher Information (FI). 
We consider the distribution $\mu_\theta(x)$. The FI of this distribution for the parameter $\theta$ at the point $\theta_0$ is  $\cfi{\mu_{\theta}(x)}{\theta}\vert_{\theta=\theta_0}$. If we now consider a new parameter $g \equiv g(\theta)$, then the FI for the same probability distribution at $g_0\equiv g(\theta_0)$  is given by
 \begin{equation}
 \cfi{\mu_{\theta}(x)}{g}\vert_{g=g_0}=\frac{\cfi{\mu_{\theta}(x)}{\theta}\vert_{\theta=\theta_0}}{\left( \frp{g(\theta)}{\theta} \vert_{\theta=\theta_0}\right)^2}\;.
 \end{equation}
In the case of a shift of the parameter, say $g(\theta)=\theta + y$ we
get $\frp{g(\theta)}{\theta}=1$ and the FI 
 is not changed. This
corresponds to the intuitive picture that by just changing what we
consider to be the parameter in the probability distribution we do not
really gain more information.  
 
 Fundamentally different is the change of the probability
 distribution. Consider the new probability distribution
 $\tilde{\mu}_\theta(x) = \mu_{f(\theta)}(x)$. Then the FI of the new
 distribution, still for the same parameter $\theta$,  
 is given by 
\begin{equation} 
 \cfi{\mu_{f(\theta)}(x)}{\theta}\vert_{\theta=\theta_0} = \left(
   \frp{f(\theta)}{\theta} \bigg\vert_{\theta=\theta_0}\right)^2
 \cfi{\mu_{\theta}(x)}{\theta}\vert_{\theta=f(\theta_0)}\;. 
\end{equation}
In the specific case of a shift of the parameter in the probability distribution, $f(\theta,\beta)=\theta+\beta$ we obtain \begin{equation}
 \cfi{\mu_{\theta+\beta}(x)}{\theta}\vert_{\theta=\theta_0} = \cfi{\mu_{\theta}(x)}{\theta}\vert_{\theta=\theta_0 +\beta}\;.
\end{equation}
This means that when we use the new distribution $\tilde{\mu}$, the
resulting  FI 
at $\theta_0$ equals the  FI 
 of the original
distribution $\mu$ at the point $\theta_0+\beta$ (see Figure
\ref{fig0}). 
\begin{figure}
\centering\includegraphics[scale=0.75]{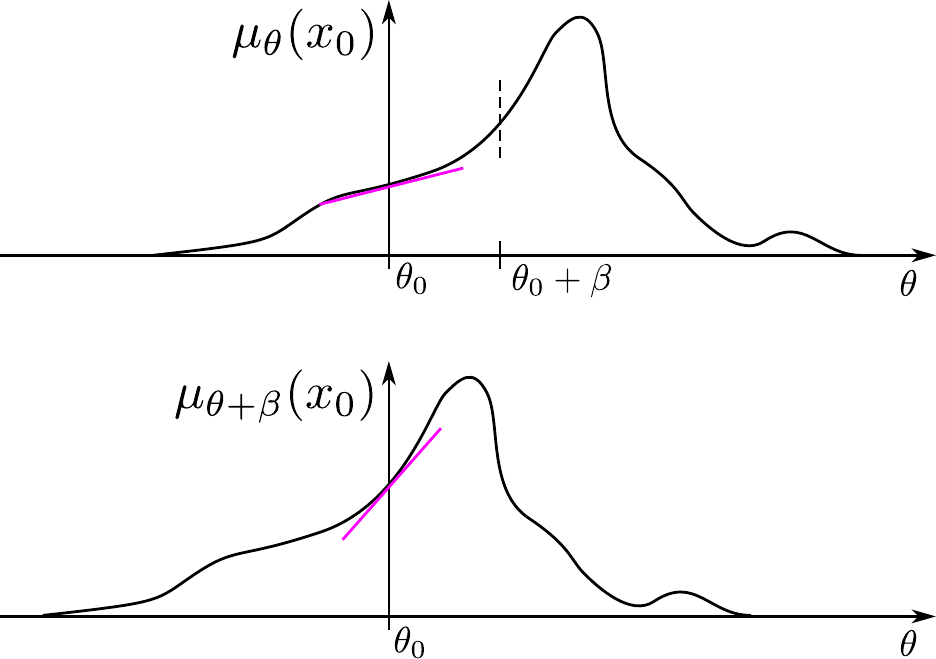}
\caption{Change of probability distribution which amounts to 
a shift of the parameter in the FI. While working at the same value of
the parameter, the second distribution (bottom plot) offers an
increased QFI. }\label{fig0} 
\end{figure}

If we now go back to physics, and if we consider $\beta$
to be a free parameter, then we can tune it to work at the most
favorable position of the Hamiltonian in terms of
$\theta$.
In the case of the broken phase shift, the most favorable values of
$\theta$ are the large values of $\theta$, where the effect of the
term $ F$ becomes negligible. Then, whatever is the original value of
$\theta_0$ we can, by adding the Hamiltonian $\beta G$ and taking
$\beta$ arbitrarily large, approach arbitrarily the upper bound of the QFI.

\subsection{Maximum sensitivity by  "Signal flooding"} 

We have seen that in order to saturate the upper bound of the channel
QFI for the "broken phase shift", $K(\theta) = \theta G +  F$,
we can add a
Hamiltonian proportional to $G $, the generator of the
original phase shift. This generator corresponds to the first
derivative of the total Hamiltonian $G=\dot{K}(\theta) $. We show in
this section that this 
is actually a general result: By adding a term proportional to the
first derivative of the Hamiltonian, we can bring the channel QFI
arbitrarily close to its upper bound. We denote the extended
Hamiltonian obtained in this way  as 
\begin{equation}
H_{\mathrm{fl}}(\theta,\beta)=H(\theta) +\beta \dot{H}(\theta_0) \;.
\end{equation} 
Instead of working directly with the channel QFI we 
first consider
the QFI for an arbitrary pure state. 
Starting with an initial state $\ket{\psi_0}$ the QFI is given by
\begin{equation}
\qfi{U_{H_{\mathrm{fl}}} \ket{\psi_0}}{\theta} =\Vstate{\msc{H}_{\mathrm{fl}}}{\ket{\psi_0}}\;,
\end{equation}
with the local generator $\msc{H}_{\mathrm{fl}}=\ii
U_{\mathrm{fl}}{}^\dagger {\dot{U}_{\mathrm{fl}}}$ 
and the evolution operator  $U_{\mathrm{fl}}=\e{-\ii H_{\mathrm{fl}}(\theta,\beta)}$. 
 The important part in this expression is the derivative of the
 evolution operator. Using eq.\eqref{eq:der_exp} we can write it as   
\begin{align*}
\dot{U}_{\mathrm{fl}}\vert_{\theta=\theta_0}&=\frp{U_{\mathrm{fl}}}{\theta}\bigg\vert_{\theta=\theta_0} \nonumber \\
&= - \ii t \int_{0}^1 \e{- \ii \alpha t H_{\mathrm{fl}}(\theta_0,\beta)}\dot{H}(\theta_0)  \e{ -\ii (1-\alpha)  t H_{\mathrm{fl}}(\theta_0,\beta) }\diff \alpha \;.
\end{align*}
The derivative of the evolution operator with respect to $\beta$ is independent of the value of $\beta$ and equals
\begin{multline}
\frp{U_{\mathrm{fl}}\vert_{\theta=\theta_0}}{\beta}= - \ii t \int_{0}^1 \e{- \ii \alpha t H_{\mathrm{fl}}(\theta_0,\beta)} \\ \times \frp{H_{\mathrm{fl}}(\theta_0,\beta)}{\beta}   \e{ -\ii (1-\alpha)  t H_{\mathrm{fl}}(\theta_0,\beta) }\diff \alpha \;,
\end{multline}

and since $\frp{H_{\mathrm{fl}}(\theta_0,\beta)}{\beta} =\dot{H}(\theta_0)$, we have 
\begin{equation}
\frp{U_{\mathrm{fl}}}{\theta}\bigg\vert_{\theta=\theta_0}=\frp{U_{\mathrm{fl}}\vert_{\theta=\theta_0}}{\beta}\;. 
\end{equation}
 We thus have 
\begin{equation}\label{eq:equality_theta_beta}
\qfi{U_{\mathrm{fl}} \ket{\psi_0}}{\theta} \vert_{\theta=\theta_0}= \qfi{U_{\mathrm{fl}}\ket{\psi_0}}{\beta}\vert_{\theta=\theta_0}\;,
\end{equation} 
 \emph{i.e.}~the QFI for $\theta$ at $\theta_0$ is equal to the QFI for $\beta$ at $\theta=\theta_0$. 
In the limit of large $\beta$ we have {$H_{\mathrm{fl}}(\theta,\beta)\vert_{\beta \gg 1} \simeq
\beta \dot{H}(\theta_0)$}. 
This Hamiltonian is, with respect to $\beta$, a phase shift
Hamiltonian which implies that
$\cqfi{H_{\mathrm{fl}}(\theta,\beta)\vert_{\beta \gg 1}}{\beta} \simeq t^2
\pnorm{\dot{H}(\theta_0)}{sn}^2$. Since eq.\eqref{eq:equality_theta_beta}
is true for 
all states $\ket{\psi_0}$, it is also true for the channel QFI, {\em
  i.e.}~{$\cqfi{H_{\mathrm{fl}}(\theta,\beta)|_{\beta\gg 
  1}}{\theta}=\cqfi{H_{\mathrm{fl}}(\theta,\beta)|_{\beta\gg 1}}{\beta}\simeq t^2
\pnorm{\dot{H}(\theta_0)}{sn}^2$}, 
showing
\emph{in fine} that ``signal flooding'' allows one 
to saturate the upper bound for the channel
QFI, 
\begin{equation}\label{eq:saturating_bound}
\cqfi{H_{\mathrm{fl}}(\theta,\beta)\vert_{\beta \gg 1}}{\theta}\vert_{\theta=\theta_0} \simeq t^2  \pnorm{\dot{H}(\theta_0)}{sn}^2\;.
\end{equation}

The interpretation of this result 
follows from Lemma \ref{sat}: by adding a large term proportional to $
\dot{H}(\theta_0)$ to $H(\theta)$, 
we bring the eigenvectors of 
$H_{\mathrm{fl}}(\theta)$ close to 
those
 of $\dot{H}(\theta_0)$. When the added
part 
dominates completely the Hamiltonian the conditions for equality given
in the Lemma  \ref{sat} are fulfilled, and the bound is saturated.
Of course, for applying this method, in general we 
need to know  the value of 
$\theta$ already in order to be able to add the Hamiltonian $\beta
\dot{H}(\theta_0)$. But the situation is not worse than what one
encounters when optimizing the POVM, which through $L_\theta$ usually
also depends  on $\theta_0$: 
The framework of the QFI and its operational
meaning are local, but
the QFI is still a useful quantity. One typically assumes that we
know already "roughly" the value of $\theta$ and that knowledge can be
used to find a near-optimal POVM.  In the present context, we would
also use this prior knowledge to determine the Hamiltonian $\beta
\dot{H}(\theta_0)$ to be added. { Moreover, for the physically important
case of a broken phase shift, $\dot{H}(\theta)$ is independent of $\theta_0$,
and hence no knowledge at all of $\theta_0$ is required for flooding
  the signal}. 
We emphasize that the Hamiltonian we add does not
depend on $\theta$, but only on $\theta_0$. Adding a Hamiltonian that
depends on $\theta$ may for sure increase the channel QFI to values
actually larger than the upper bound, but requires not only prior
information, but also the need to design a way to add an extra
dependence on the 
parameter.  This would be comparable to  
adding an extra dependence on the parameter through a
$\theta$-dependent POVM \cite{seveso_quantum_2017}.

\subsection{Subtracting the Hamiltonian}

The first method discussed in section \ref{sec:restoring} in the context of a
broken phase shift, namely subtracting the disturbing part
$F=K(\theta_0)- \theta_0 \dot{K}(\theta_0)$ from the
Hamiltonian $K(\theta)$ 
can be generalized 
further: We can
subtract the entire Hamiltonian at $\theta_0$ from the Hamiltonian,
leading to a new Hamiltonian 
\begin{equation}
H_{\mathrm{sub}}(\theta) = H(\theta) - H(\theta_0)\;.
\end{equation} 
This is a valid Hamiltonian extension in the sense that we added a
$\theta$-independent operator to the Hamiltonian while keeping its
parametric derivative, 
$\dot{H}_{\mathrm{sub}}(\theta)=\dot{H}(\theta)$. 
Moreover, at  $\theta=\theta_0$ this Hamiltonian vanishes,
$H_{\mathrm{sub}}(\theta_0) =0$, 
 and 
therefore commutes with any operator, in particular with its own
derivative, $\com{H_{\mathrm{sub}}(\theta_0)
}{\dot{H}_{\mathrm{sub}}(\theta_0) }=0$.  This implies 
that we thus saturate the bound: 
 \begin{equation}
\cqfi{H_{\mathrm{sub}}(\theta) }{\theta} \vert_{\theta=\theta_0}= t^2 \pnorm{\dot{H}(\theta_0)}{sn}^2\;.
\end{equation} 
  In full generality we can subtract $H(\theta_0)- Q$, with
  $\com{Q}{\dot{H}(\theta_0)}=0$ and still saturate the bound. Below
  we show for the example of the measurement of the direction of the
  magnetic field how a locally vanishing Hamiltonian can be realized.





 We now check how stable the method is if one does not subtract exactly
 $\dot{H}(\theta_0)$ but rather $\dot{H}(\theta_0+\ve)$. We define the
 extended Hamiltonian $H_{\mathrm{sub},\ve}(\theta)\equiv H(\theta) -
 H(\theta_0+\ve)$. To second order in $\ve$ we have  
 \begin{equation}
H_{\mathrm{sub},\ve}(\theta) = H(\theta)- H(\theta_0) -\ve \dot{H}(\theta_0)-\ve^2 \ddot{H}(\theta_0)/2 +\mc{O}(\ve^3)\;.
\end{equation} 
To obtain the channel QFI we need the eigenvalues of $\msc{H}_{\mathrm{sub},\ve}$, 
the
local generator  corresponding to $H_{\mathrm{sub},\ve}(\theta)$. 
Using $  e^{X}Ye^{-X}=Y+\left[X,Y\right]+{\frac {1}{2!}}[X,[X,Y]]+\cdots $
in eq.\eqref{eq:der_exp} 
and the fact that $H_{\mathrm{sub},\ve}(\theta_0)
= -\ve \dot{H}(\theta_0)-\ve^2 \ddot{H}(\theta_0)/2 +\mc{O}(\ve^3)$, we obtain, up to second order,
\begin{equation}
\msc{H}_{\mathrm{sub},\ve} = t( \dot{H}(\theta_0) - \ii \frac{\ve^2 t}{2} \Gamma +\mc{O}(\ve^3) ) \;,
\end{equation}
where $\Gamma=[\ddot{H}(\theta_0,\dot{H}(\theta_0)]/2$. 
We can now use 
perturbation theory to see how the eigenvalues of $ \msc{H}_{\mathrm{sub},\ve}$ are
affected by $\ve$.  In its eigenbasis, $\dot{H}(\theta_0)$ is written
$\dot{H}(\theta_0)=\sum_{i=1}^d e_i \ddens{i}$. We denote the
eigenvalues of $\msc{H}_{\mathrm{sub},\ve}$ by  $e_i^{(\ve)}$. Assuming
non-degenerate $e_i^{(\ve)}$, we have to first order perturbation
theory in $\ve^2$ 
\begin{equation}
e_i^{(\ve)}= t e_i-\ii\frac{ \ve^2t^2}{2}\moyvec{\Gamma}{i} +\mc{O}(\ve^3)\;.
\end{equation}
Provided that $\ve$ is small enough, no new degeneracies will
appear. If $e_1$ and $e_d$ are respectively the maximal and minimal
eigenvalue of $\dot{H}(\theta_0)$, then the channel QFI of the
extended channel is given up to second order in $\ve$ by
\begin{multline}
\cqfi{H_{\mathrm{sub},\ve}(\theta)}{\theta}=t^2 \pnorm{\dot{H}(\theta_0)}{sn}^2 \\-\ii \ve^2 t^3(e_1-e_d)(\moyvec{\Gamma}{1}-\moyvec{\Gamma}{d})\;.
\end{multline}
This shows that errors of the order $\epsilon$ in the value of
$\theta_0$ lead to a channel QFI reaching the upper bound up to
a correction of order $\ve^2$. 

   { It was shown in \cite{yuan_optimal_2015} that one can also saturate the upper bound with a different method: by breaking the
   evolution operator $U_{H(\theta)}$ in $m$ evolution operators with
   evolution time $\tau=t/m$, and interspersing controls
   $U_{\mathrm{c}}$ one can increase the channel QFI. An optimal
   choice of controls  leads to the saturation of the upper bound. }
Neither of the three methods 
signal flooding, using additional interspersed controls, and
Hamiltonian subtraction 
makes any use of
ancillas. This shows that in all generality, ancillas are not
\emph{necessary} to achieve the maximal 
sensitivity when estimating a
Hamiltonian parameter. While this result was already known for phase
shift Hamiltonians
\cite{boixo_generalized_2007,fraisse_hamiltonian_2016}, 
these
methods show that it is the case for 
{\em any} Hamiltonian.

\subsection{Engineering the time dependence of the channel QFI}\label{sec:scaling}

Hamiltonian extension can also be used  to modify the behavior of the
channel QFI with time or other relevant resources.
  It was shown in  \cite{pang_quantum_2014,pang_erratum:_2016} that 
in case the eigenvalues of the Hamiltonian do not depend on the
parameter to be estimated,
the QFI behaves periodically in $t$
(this discussion on $t$ applies actually to any 
phase shift parameter). In general for $H(\theta)=\sum_{i=1}^d \lambda_i(\theta) \ddens{\psi_i(\theta)}$ (we assume for simplicity the Hamiltonian non-degenerate) the local generator is given \cite{pang_quantum_2014} by 
  \begin{multline}
  \msc{H}= t \sum_{i=1}^d \frp{ \lambda_i(\theta)}{\theta} \ddens{\psi_i(\theta)}+ 2\sum_{k \neq l} \e{-\ii t \frac{ \lambda_k(\theta)- \lambda_l(\theta)}{2}}\\ \times \sin( \frac{ \lambda_k(\theta)- \lambda_l(\theta)}{2}) \pscal{\psi_l(\theta)}{\partial_\theta \psi_k(\theta)} \dens{\psi_l(\theta)}{\psi_k(\theta)}\;,
  \end{multline}
  clearly showing that if $ \frp{ \lambda_i(\theta)}{\theta}=0$,
  \emph{i.e} if the eigenvalues of $H(\theta)$ are
  $\theta$-independent, only the periodic term survives. The fact that
  the channel QFI (and more generally the QFI) behaves only
  periodically with the time prohibits the use of time as a
  resource: We cannot  increase the sensitivity to arbitrarily large
  values by increasing the evolution time. This is particularly  
harmful since quantum metrology typically provides a quadratic scaling with time
(for time-independent Hamiltonians), to be compared to the linear
scaling obtained by classical averaging. 

To show how Hamiltonian extension can help to engineer the time
dependence we assume that the  eigenvalues of the original Hamiltonian
$H(\theta)$ do not depend on $\theta$: $\lambda_i(\theta) \rightarrow
\lambda_i$. We then consider the Hamiltonian extension  
\begin{equation}
H_V(\theta)=H(\theta)+ \ve V\;.
\end{equation}
where $V$  can be any operator. For small $\ve$ we can use the
perturbation theory to first order to find the perturbed eigenvalues
$\lambda_i^{\ve V}$ of $H_V(\theta)$. Assuming non-degenerate 
$\lambda_i$, we have 
\begin{equation}
\lambda_i^{\ve V}=\lambda_i +\ve \moyvec{V}{\psi_i(\theta)}\;.
\end{equation}
We see that the introduction of $\epsilon V$ in the Hamiltonian makes
the eigenvalue depend on $\theta$,  as long as
$\moyvec{V}{\psi_i(\theta)}$ is not constant as function of $\theta$. 
Under this
condition  this Hamiltonian extension   
introduces a quadratic
scaling  with time in the QFI. Despite the fact that for a general $V$
this Hamiltonian extension does typically not allow one to saturate the upper bound
it offers the advantage that one does not need  to know the exact
value of the parameter to implement it.

{  The method of  engineering the time dependence shows that in general,
not only
$\dot{H}(\theta_0)$ can be used to increase the channel QFI. One can
check case by case if adding another operator helps to increase the
best sensitivity. 
For example it was shown for a broken phase shift
  of the form $K=\theta G + \eta F$ that the channel
  QFI is not always a monotone function of $\eta$  \cite{de_pasquale_quantum_2013}. Thus, for certain values of $\eta$ the channel QFI can be increased 
  by increasing $\eta$, an effect
  that the authors call ``dithering''.}

\section{Example of applications}\label{sec:NV}

\subsection{NV center magnetometry}

 The nitrogen-vacancy defects in diamonds, also known as NV centers,
 correspond to defects in a diamond crystal lattice, where a
 substitutional nitrogen atom comes 
with a vacancy in one of the
 neighbouring sites. Such NV centers exist in three forms, a neutral
 one, a positively charged one, and a  negatively charged one. The
 latter provides a promising system for magnetometry, since it has a
 spin triplet which can be monitored efficiently through optical
 processes, and offers a coherence time that can be as high as a few
 ms (see
 \cite{schirhagl_nitrogen-vacancy_2014,rondin_magnetometry_2014} for
 recent reviews). 
 
Neglecting the interactions with the $^{14}$N nuclear spin as 
well as
the bath of the $^{13}$C nuclear spins, the Hamiltonian $H_{\mathrm{NV}}$  for the triplet
state of the NV center can be written \cite{rondin_magnetometry_2014} 
\begin{equation}\label{eq:ham_NV}
 H_{\mathrm{NV}}= g \mu_\mathrm{B} ( B_x S_x +B_y S_y + B_z
 S_z)/\hbar+ D S_z^2/\hbar + E (S_x^2-S_y^2)/\hbar\;, 
 \end{equation}
 with $g$ the Land\'e factor, $\mu_\mathrm{B}$
the Bohr magneton, $D$ and $E$ the zero field splitting parameters and
$S_x, S_y$ and $S_z$ the dimensionless spin-1 matrices, fulfilling
$\com{S_i}{S_j}=\ii \hbar \ve_{ijk} 
\,S_k \;\forall \;i,j,k \in \lbrace
x,y,z \rbrace$. The zero field splitting 
has two components, the axial one, with parameter $D$  (taken as
$D=2\pi\times 2.87\,$GHz), and the off-axis one, with parameter $E$ (taken as
$E=2\pi\times 5\,$MHz). The parameter that one seeks to estimate is $B_z$, the
magnetic field in the $z$ direction which is the direction of  the
axis from N to V. 
This Hamiltonian has the form of a "broken phase shift"
$H_{\mathrm{NV}}=B_z G +F$ with $G=g \mu_\mathrm{B} S_z/\hbar$. 
The channel QFI is  bounded by $(t/\hbar)^2\pnorm{G}{sn}^2 =(t g
\mu_\mathrm{B} /\hbar )^2 \pnorm{S_z/\hbar}{sn}^2 
= 4(t g \mu_\mathrm{B} /\hbar)^2$. 
 Due to
the part which does not commute with $S_z$, \emph{i.e.}~the magnetic
transverse field and the off-axis zero field splitting, the channel
QFI decreases for small value of $B_z$, while for high values of
$B_z$ the channel QFI reaches the upper bound (see Figure
\ref{fig1}).

As we have seen in section \ref{sec:saturating_bound}, by adding the
derivative $\dot{H}_{\mathrm{NV}}$ of the Hamiltonian to the original
Hamiltonian we can saturate the upper bound.  The shifted Hamiltonian is  
\begin{equation}\label{eq:Ham_NV_ext}
\tilde{H}_{\mathrm{NV}}=H_{\mathrm{NV}} +\beta
\dot{H}_{\mathrm{NV}}=H_{\mathrm{NV}} + \beta g \mu_B S_Z/\hbar\;. 
\end{equation}
We  represent in Figure~\ref{fig1} the effect of $\beta$ on the
channel QFI. We see that by increasing the value of $\beta$ we {can get
arbitrarily close} 
to the upper bound.  
As  pointed out when
discussing the "broken phase shift", the effect of $\beta$ amounts
here to evaluating the QFI at shifted value of the parameter
(at $\theta+\beta$ instead of at $\theta$).
For magnetometry with NV-centers {it is actually already known that
  adding an additional magnetic field 
can help to measure weak magnetic
fields}.  In \cite{rondin_magnetometry_2014} this was discussed in the
the context of reaching a linear Zeeman effect. In addition,  
adding a bias field makes already sense from the perspective of
shifting the precession signal up to higher frequency, where it can be
distinguished from noise more easily.  Here we see that independently
of such specific considerations, ``flooding the signal'' 
by adding the known parameter-derivative of the Hamiltonian with a
large factor is a very general method that allows one to overcome
pernicious effects of other parts of the
Hamiltonian and reach maximal possible sensitivity.

\begin{figure}
\includegraphics[scale=0.25]{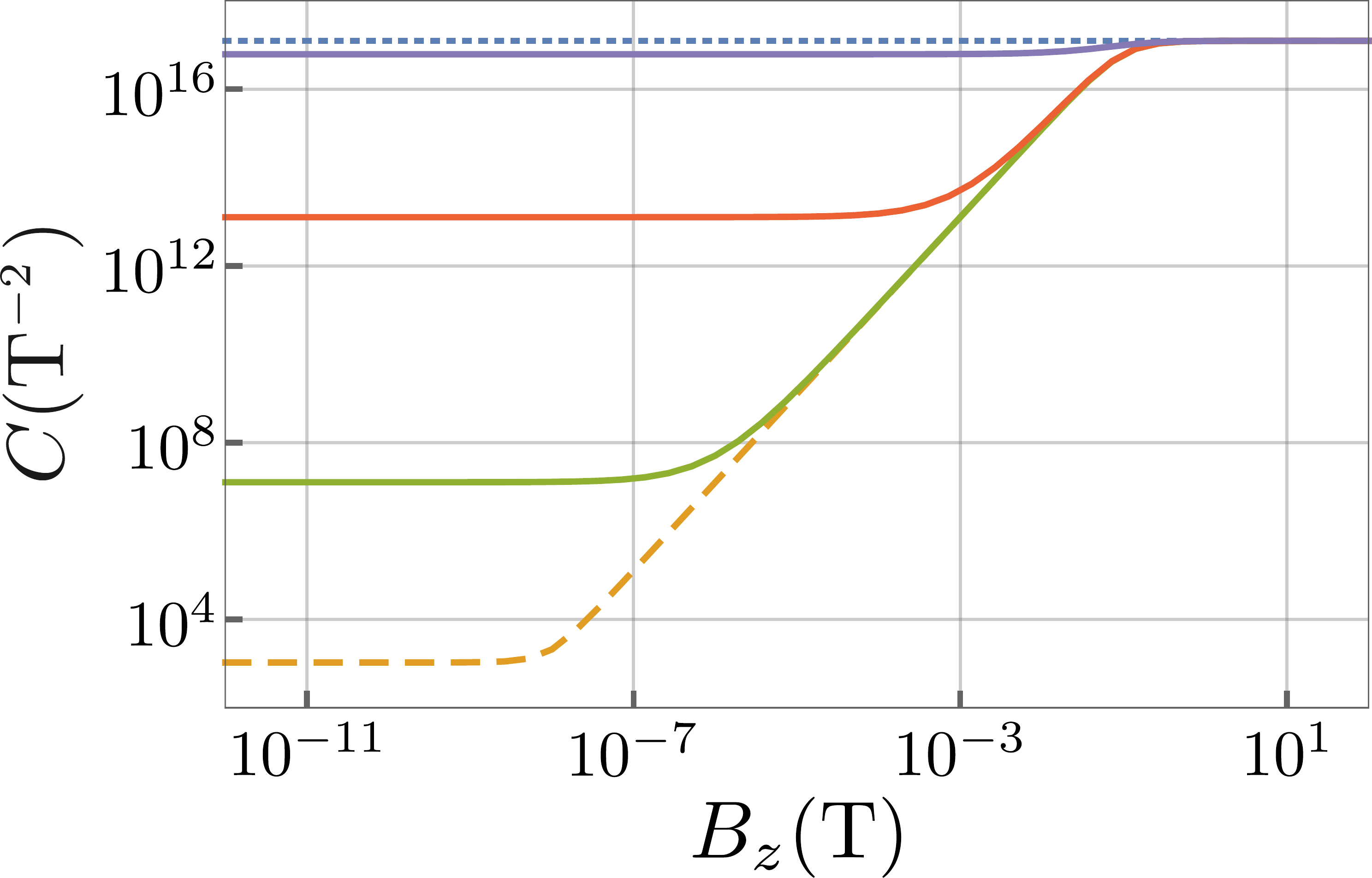}
\caption{Channel QFI for the NV center
  ($t=10^{-3}$s, $B_x=10^{-1}$T and $B_y=0$T), as a function of the
  parameter to be estimated, $B_z$. The dotted line represents the
  upper bound, the dashed line  the channel QFI of the original Hamiltonian \eqref{eq:ham_NV} and the
  continuous line the channel QFI of the extended channel
  \eqref{eq:Ham_NV_ext} for different values of $\beta$: From bottom to top we have $\beta=10^{-6}$, $\beta=10^{-3}$ and
  $\beta=10^{-1}$. }\label{fig1}  
\end{figure}

The behaviour of the channel QFI for very low values of $B_z$ deserves
some more comments. As it can be seen in Figure \ref{fig1} when $B_z$
becomes very small the channel QFI reaches a plateau. This is a quite
general feature for Hamiltonians of the form of a "broken phase shift"
$K(\theta)=\theta G + F$. The channel QFI $\cqfi{K(\theta)}
{\theta}\vert_{\theta=0}$ 
 at $\theta=0$ can be obtained by calculating
 $\msc{K}\vert_{\theta=0}$ which reads 
\begin{equation}\label{eq:local_generator_theta_0}
\msc{K}\vert_{\theta=0}=t\sum_{i}g_{ii}\ddens{i}+\ii \sum_{i\neq j} g_{ij} \frac{1-\e{\ii  t ( f_i-f_j)}}{f_i-f_j} \dens{i}{j}\;,
\end{equation}
with $ F = \sum_i f_i \ddens{i}$, where the $f_i$s and the
$\ket{i}$s are respectively the eigenvalues (assumed non-degenerate
here) 
and eigenvectors of $F$,
and $G=\sum_{i,j} g_{ij} \dens{i}{j}$, with $g_{ij}^*=g_{ji}$. To get
more insight into this expression 
notice that the channel QFI vanishes only when maximal and minimal
eigenvalues of $\msc{K}$ coincide, 
{\em i.e.}~when $\msc{K}$ 
is proportional to the
identity operator. In general this will be the case 
if $t(f_i-f_j) $ is an 
integer and $g_{11}=g_{dd}$. 
This condition 
is not necessary 
since $G$
 can be
sparse and therefore some $g_{ij}$ may already be equal to zero. Still,
this simple analysis shows that particular cases excepted,
it is a
quite general feature that the channel QFI for a broken phase shift
does not vanish for 
small values of the parameter. 

\subsection{Estimation of a direction of a magnetic field}

The estimation of a component 
of a magnetic field using a NV center studied in
the previous section 
leads to a broken phase shift. We now consider 
the estimation of one of the spherical angles characterizing the direction of a
magnetic field with a free spin-1 as a probe, a situation
that does not correspond to a broken 
phase shift. The Hamiltonian
is given by 
\begin{equation}\label{eq:Ham_direction}
H(B,\theta,\vp) =g \mu_{\mathrm{B}} \mathbf{B}\cdot
\mathbf{S}/\hbar\;,
 \end{equation}
 with $ \mathbf{B}= B(\sin(\theta) \cos(\vp),\sin(\theta) \sin(\vp),
 \cos(\theta)  )$ 
and $\mathbf{S}=(S_x,S_y,S_z)$. 
  We want to estimate the parameter $\theta$ in a scalar parameter
  setting (we consider $B$ and $\vp$ as known). The channel QFI for
  the corresponding channel is bounded as
 \begin{equation}
 \cqfi{H(B,\theta,\vp)}{\theta}\leq \left(t / \hbar\right)^2
 \pnorm{\frp{H(B,\theta,\vp)}{\theta}}{sn}^2\;. \label{39}
 \end{equation}
The eigenvalues of $H(B,\theta,\vp) $ are $0$ and $\pm g
\mu_{\mathrm{B}}  B$, from 
which we get  $\pnorm{\frp{H(B,\theta,\vp)}{\theta}}{sn}^2 = 4 (g
\mu_{\mathrm{B}}  B)^2$. The local generator  $\msc{H}_\theta$ of the
translation in $\theta$ can be computed exactly and its eigenvalues
are $0$ and $\pm 2 \sin(g \mu_{\mathrm{B}}Bt/(2\hbar))$. This gives a
channel QFI equal to  
\begin{equation}
 \cqfi{H(B,\theta,\vp)}{\theta}=16 \sin^2( g \mu_{\mathrm{B}}Bt/(2\hbar)) \;.
\end{equation}
We see that the channel QFI has a periodic time dependence since
the
eigenvalues of $H(B,\theta,\vp) $ are $\theta$-independent (see
dashed line in Figure \ref{fig:direction_plot}). However, the upper
bound still scales quadratically with time. As a result, for large
time, the discrepancy between the actual channel QFI and its upper
bound increases. Notice that for this Hamiltonian the role of 
time $t$ in the channel QFI is the same as the role of the strength $B$ of
the magnetic field. 

\subsubsection{Signal flooding}

We now examine
how "signal flooding" helps to increase the channel QFI. We consider the extended Hamiltonian
\begin{equation}\label{eq:Ham_sig_flod_direction}
H_{\mathrm{fl}}(B,\theta,\vp) = H(B,\theta,\vp) +\beta \frp{H(B,\theta,\vp)}{\theta}\bigg\vert_{\theta=\theta_0}\;.
\end{equation}
We represent in the Figure \ref{fig:direction_plot} the effect of
signal flooding by plotting the channel QFI of
$H_{\mathrm{fl}}(B,\theta,\vp)$ for different values of $\beta$. In the same
way as for the NV center Hamiltonian we see that increasing the value
of $\beta$ allows one to get arbitrarily close to the upper bound. 
In contrast to
the estimation of a broken phase shift (\emph{e.g.}~estimation of $B_z$ for the NV center),  here signal flooding does not
correspond to a
shift of the value of $\theta$.

\begin{figure}
\centering\includegraphics[scale=0.25]{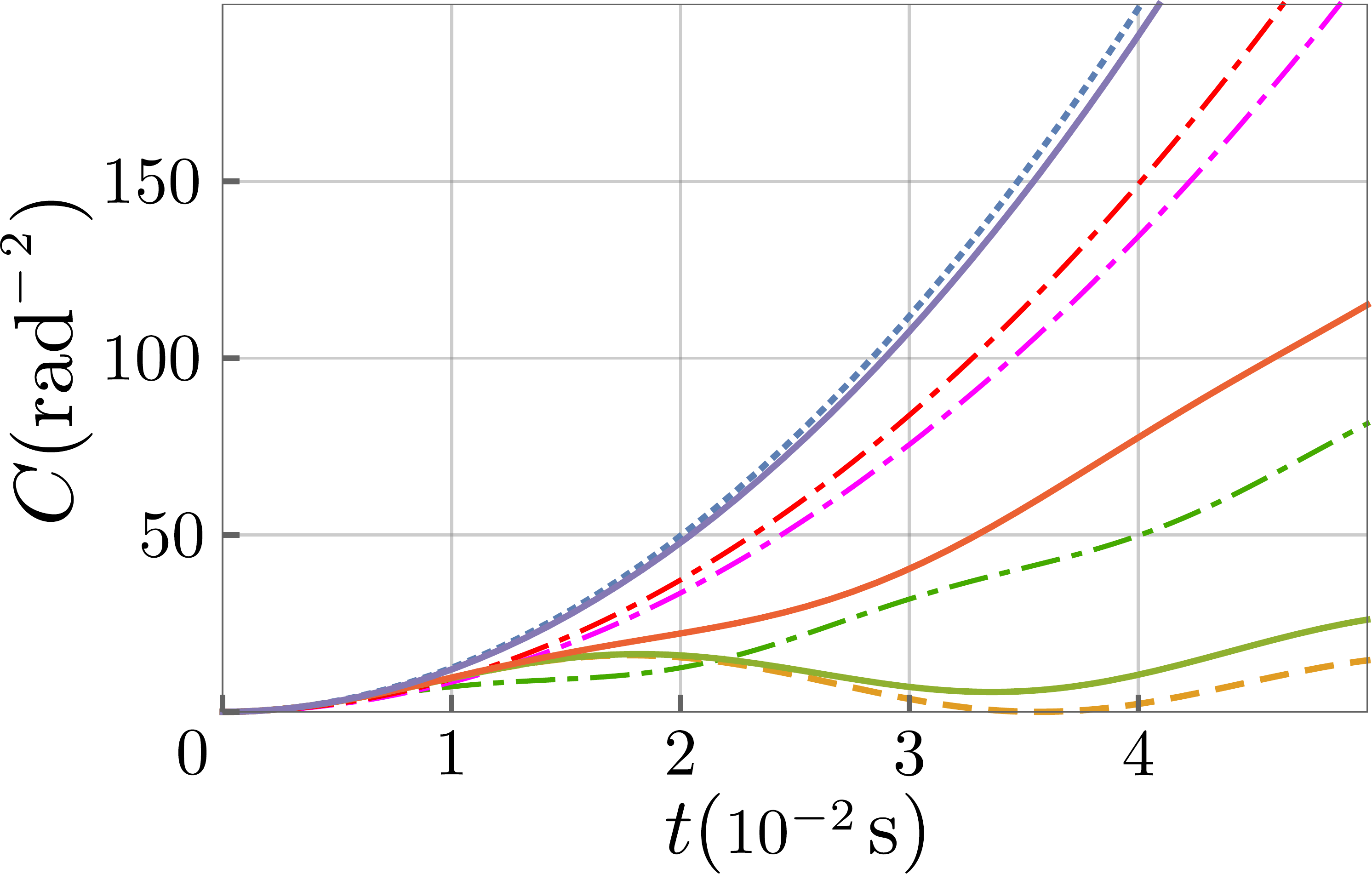}
\caption{Channel QFI for a direction of the magnetic field ($B=10^{-9}$T, $\vp=\pi /4$ and $\theta=\pi /3$), as a function of the time. The dotted line represents the
  upper bound, the dashed line  the channel QFI of the original Hamiltonian \eqref{eq:Ham_direction}. The
  continuous 
line represents the channel QFI for the "signal flooding" Hamiltonian \eqref{eq:Ham_sig_flod_direction} for different values of $\beta$: From bottom to top we have $\beta=0.2$, $\beta=0.75$ and
  $\beta=5$. The dotted-dashed line represents the channel QFI for the extended Hamiltonian \eqref{eq:Ham_quad}, for different values of $\kappa$: From bottom to top we have $\kappa=1$, $\kappa=10$ and $\kappa=10^9$. }\label{fig:direction_plot}  
\end{figure}

\subsubsection{Time engineering}

Having eigenvalues independent of $\theta$, $H(B,\theta,\vp)$ leads to a channel QFI with a  periodic time scaling. In line with section \ref{sec:scaling} we now show how we can restore the quadratic time scaling using a Hamiltonian extension. We consider the extended Hamiltonian
\begin{equation}\label{eq:Ham_quad}
H_{S_z} (B,\theta,\vp)=H(B,\theta,\vp) +\kappa B g \mu_{\mathrm{B}}
S_z/\hbar\;,
\end{equation}
which correspond to the original Hamiltonian with an additional
magnetic field in the $z$ direction with a strength $\kappa B$.

The eigenvalues of $H_{S_z} (B,\theta,\vp)$ are $0$ and $\pm g
\mu_{\mathrm{B}} \sqrt{B^2+\kappa^2 +2B \kappa \cos(\theta)}$, showing
that the additional magnetic field makes the eigenvalues depend 
 on $\theta$, and therefore allow one 
to create a quadratic time scaling. We have represented in Figure
\ref{fig:direction_plot} the channel QFI of the extended Hamiltonian
for different values of $\kappa$. This clearly shows that the
additional magnetic field introduces a quadratic scaling, but also
that the larger $\kappa$, \emph{i.e.}~the strength of the
additional field, the larger is the channel QFI. In contrast
to signal flooding, for very large values of $\kappa$ we do not
reach the upper bound (compare the case $\beta = 5$ and $\kappa=
10^9$): Indeed $H_{S_z} (B,\theta,\vp)$ do not fulfil the condition
of Lemma \ref{sat} and when it dominates completely the Hamiltonian we
do still not reach the upper bound.

\subsubsection{Hamiltonian subtraction}

\begin{figure}
\centering\includegraphics[scale=0.25]{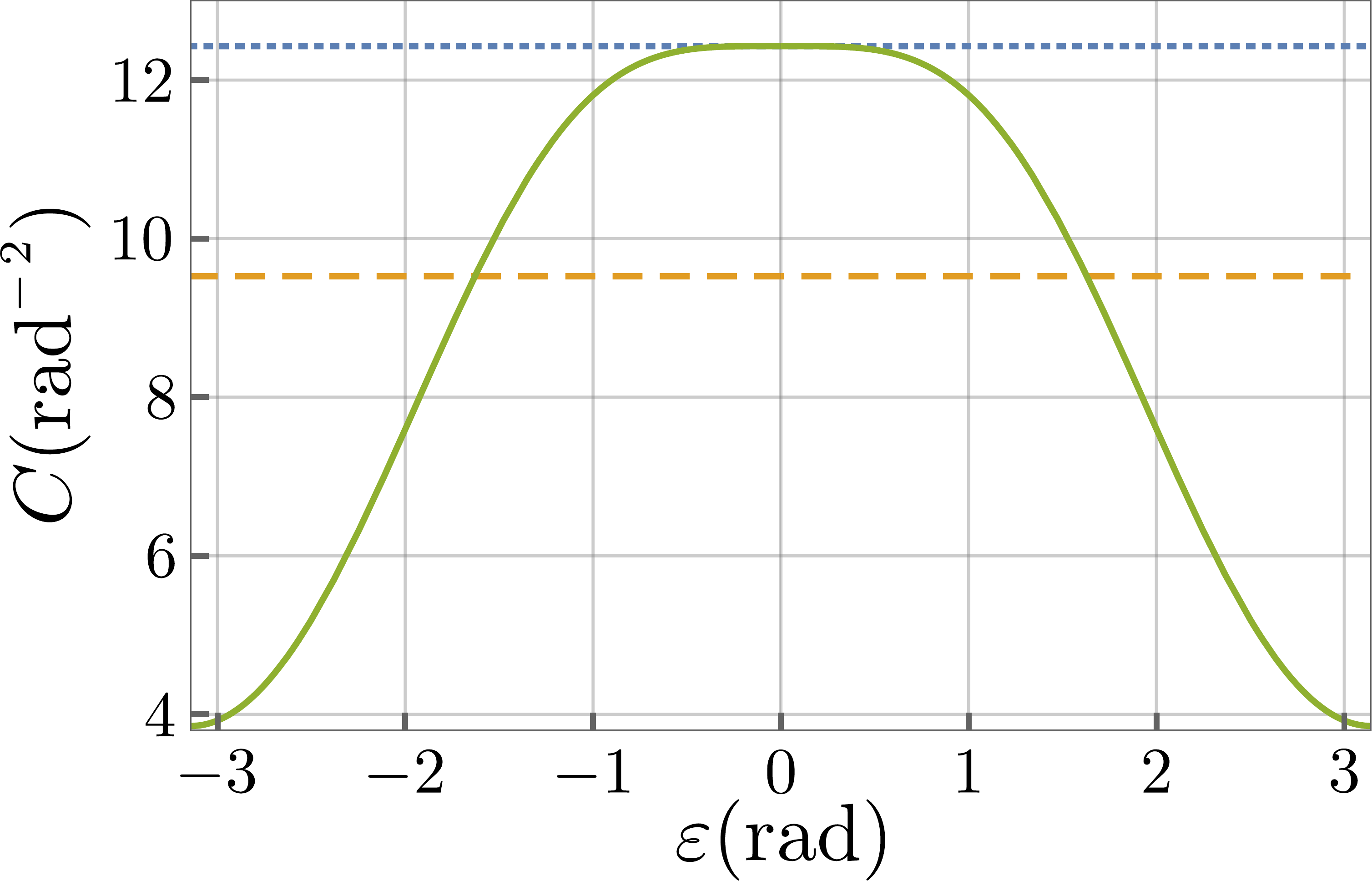}
\caption{Effect of a perturbation in "Hamiltonian subtraction" for the
  estimation of a direction of a magnetic field.  ($B=10^{-9}$T,
  $\vp=\pi /4$, $\theta_0=\pi /3$ and  $t=10^{-2}$s). The dotted line
  represent the upper bound and the dashed line the channel QFI of the
  original Hamiltonian \eqref{eq:Ham_direction}. The plain line
  represent the channel QFI of the Hamiltonian
  \eqref{eq:Ham_sub_magnetic}.}.\label{fig:sub_ham}   
\end{figure}

Finally, we show how we saturate the upper bound using Hamiltonian
subtraction. Hamiltonian subtraction here amounts to adding a magnetic
field of the same strength but 
in opposite direction to the original magnetic field. To
see the effect of a 
slight deviation from the correct direction of the added field, 
we study
the perturbed Hamiltonian 
\begin{equation}\label{eq:Ham_sub_magnetic}
H_{\mathrm{sub},\ve}(B,\theta,\phi)=H(B,\theta,\phi)-H(B,\theta_0+\ve,\phi)\;.
\end{equation}
We can 
calculate the channel QFI exactly,
\begin{multline}\label{eq:cQFI_pert}
\cqfi{H_{\mathrm{sub},\ve}(B,\theta,\phi)}{\theta}\vert_{\theta=\theta_0}=4(g  \mu_\mathrm{B}
B \frac{t}{\hbar} )^2 \cos^2(\frac{\ve}{2} )\\ + 4 \sin^2(g
\mu_\mathrm{B} B \frac{t}{\hbar}  \sin(\frac{\ve}{2}))\;. 
\end{multline}
 One verifies
 that for values of $\ve=0$ this 
saturates the upper bound \eqref{39}. 
It 
 is interesting to notice that the correction of second order in $\ve$
 in the channel QFI vanishes exactly, and the leading order corrections
 are of order $\ve^4$, demonstrating the stability of the method to
 perturbation in the direction of the subtracted magnetic field, as
 represented in Figure \ref{fig:sub_ham}. 
\section{Conclusion}  \label{ccl}

We investigated the problem of single Hamiltonian parameter
estimation, and the effect of Hamiltonian extension on the precision
with which one can estimate the parameter. It was already known that
in the case of a phase shift, Hamiltonian extension does not lead to
an increase of the channel QFI (QFI optimized over all input
states) \cite{fraisse_hamiltonian_2016,boixo_generalized_2007}. 
But for more general Hamiltonians, Hamiltonian extension 
may increase the channel QFI.  In particular, we found two ways of
engineering the Hamiltonian to saturate the upper bound of the 
channel QFI: ($i$)"Signal flooding"  which consists of adding to the
original Hamiltonian a large term proportional to its 
derivative; and $(ii)$ "Hamiltonian subtraction" which consists in
subtracting the Hamiltonian at a fixed value of the parameter from the
original Hamiltonian.

Neither method makes use of any
ancillas, showing that adding subsystems is in general not necessary
for unitary parameter estimation to achieve the best precision.  
We applied "signal flooding" to the Hamiltonian of an NV
center. Such systems are used to measure small magnetic fields. We
showed how adding a 
magnetic field in the $z$-direction 
helps to increase the maximal possible sensitivity of the
measurement of
the magnetic field's component in the same
direction. 
  We also illustrated  both methods with the estimation of a direction
  of a  magnetic field.

Finally we showed that in cases where the eigenvalues of the original
Hamiltonian are parameter-independent, {adding almost any constant
Hamiltonian} can lead to a quadratic increase of the channel QFI with
time, whereas for the original Hamiltonian it is bounded and periodic
in time. Such is the
situation for the measurement of the polar angle of the magnetic
field, probed with a free spin. This will typically not enable
saturation of the bound of the channel QFI, but can nevertheless be
very advantageous for large measurement times, and relatively straight-forward to
implement.

The scenario considered in this work is an ideal one. 
Our figure of merit, the channel QFI, constitutes a valid
and achievable benchmark for quantum parameter estimation 
that is difficult to reach
in practice. 
Even if the theory gives us the optimal POVM and
the optimal state, they may be  hard to implement.
Therefore, it may be interesting to see to what extend 
Hamiltonian extensions
also work in 
situations where the optimal POVM or the
optimal state cannot be implemented.  
\FloatBarrier
{\bf Acknowledgment} We thank Fabienne Schneiter for  useful discussions.

\bibliography{biblio_ham_ext}

\end{document}